\documentclass[conference]{IEEEtran}
\IEEEoverridecommandlockouts

\usepackage{cite}
\usepackage{amsmath,amssymb,amsfonts}
\usepackage{algorithmic}
\usepackage{graphicx}
\usepackage{textcomp}
\usepackage{xcolor}
\usepackage{amsthm}
\newtheorem{theorem}{Theorem}

\begin{document}

\title{Statistical QoS Provisioning for Underwater Magnetic Induction Communication}

\author{
	\IEEEauthorblockN{
		Zhichao Li\IEEEauthorrefmark{1}, 
		Jianyu Wang\IEEEauthorrefmark{1}, 
		Wenchi Cheng\IEEEauthorrefmark{1}, 
		and Yudong Fang\IEEEauthorrefmark{2}} 
	\IEEEauthorblockA{\IEEEauthorrefmark{1}State Key Laboratory of Integrated Services Networks, Xidian University, Xi’an, China}
	\IEEEauthorblockA{\IEEEauthorrefmark{2}Ministry of Emergency Management Big Data Center, Beijing, China}
    E-mail: \textit{24011210926@stu.xidian.edu.cn, \{wangjianyu, wccheng\}@xidian.edu.cn, fangyudong9713@ustc.edu} 
}

\maketitle

\begin{abstract}
    Magnetic induction (MI) communication, with stable channel conditions and small antenna size, is considered as a promising solution for underwater communication network. However, the narrowband nature of the MI link can cause significant delays in the network. To comprehensively ensure the timeliness and effectiveness of the MI network, in this paper we introduce a statistical quality of service (QoS) framework for MI communication, aiming to maximize the achievable rate while provisioning delay and queue-length requirements. Specifically, we employ effective capacity theory to model underwater MI communication. Based on convex optimization theory, we propose a current control strategy that maximizes the effective capacity under the constraints of limited channel capacity and limited power. Simulations demonstrate that the current control strategy proposed for MI communication differs significantly from that in the conventional statistical QoS provisioning framework. In addition, compared to other current control strategies, the proposed strategy substantially improves the achievable rate under various delay QoS requirements.
\end{abstract}


\section{Introduction}

Electromagnetic (EM) wave communication has developed significantly, with substantial improvements in achievable rate, bit error ratio, bandwidth utilization, etc. However, in underwater, dynamic underwater conditions and harsh propagation environments severely hinder the propagation of EM waves. In such environments, magnetic induction (MI) communication, utilizing the mutual inductance effect of the coils, offers advantages such as high penetration efficiency, stable channel conditions, and small antenna size\cite{5199549,10210297}. MI communication is widely used in marine resource surveys, underwater ecological monitoring and submarine communications, etc.\cite{8464253}.

\begin{figure*}[h]
    \centering
    \includegraphics[width = 7in]{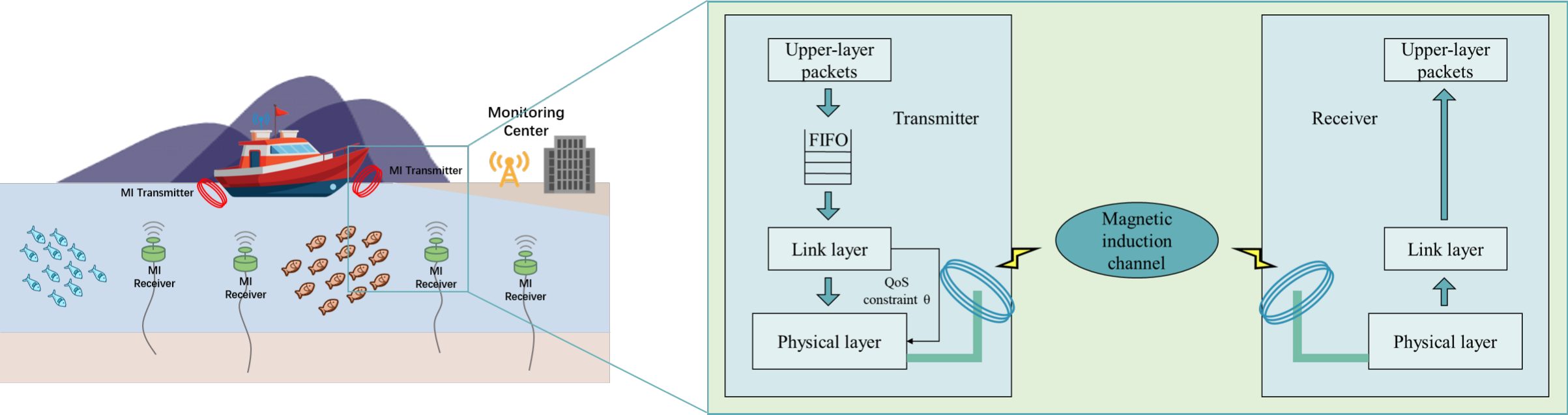}
    \caption{Underwater MI networks and statistical QoS provisioning framework with finite-length queue.}
    \label{Underwater MI networks and equivalent circuits.}
\end{figure*}

Several studies have been conducted on MI communications. The authors of \cite{5199549} comprehensively characterized the underground MI communication channel. Based on this channel model, MI waveguide technology was developed to reduce MI path loss, effectively improving the communication performance of MI systems. MI systems used for wireless communication or power transfer often experience coil misalignment due to movement, which causes the coils not to be perfectly parallel and aligned. To analyze completely misaligned links, the authors of \cite{7794767} proposed a random misalignment model, assuming coil orientations are uniformly and randomly distributed. He derived the mutual inductance probability density function (PDF) for practical configurations consisting of single coils and collinear coil arrays. Compared to terrestrial or underground environments, underwater environments exhibit more uncertainty. In the presence of water flow, the positions of the transmitter and receiver change as the communication channel varies. To address this, the authors of \cite{8851668} developed a dynamic underwater MI model and provided path loss and channel transfer functions. The authors of \cite{9681845} developed the Multi-frequency Resonating Compensation (MuReC) coil based Multiple-band, Multiple-Input, and Multiple-Output (MbMIMO) simultaneous wireless information and power transfer (SWIPT), which can generate multiple resonant frequencies with one coil. 

To enhance the system throughput and interference resistance of MI communication, many studies have explored the application of multiple input multiple output (MIMO) technology in MI systems. In near-field magnetic induction (NFMI) communication, the authors of \cite{7428884} proposed using heterogeneous antenna arrays with multiple pole antennas to overcome strong crosstalk between transmitters and enable MIMO transmission. This approach effectively eliminates crosstalk, enhancing the capacity. The authors of \cite{9580360}, proposed an optimal current control scheme in a closed form for multiple input multiple output orthogonal frequency division multiplexing (MISO-OFDM) MI communication system, which aims to achieve the maximum achievable rate under a transmit power constraint.

However, the narrowband nature of the MI link can cause significant delays in the network. To achieve efficient transmission in an MI communication system with limited delay and queue length, we present a MI communication system model with based on statistical quality of service (QoS) provisioning. It provides closed-form expressions for transmit power, receive signal-to-noise ratio (SNR), and effective capacity. We formulate the optimization problem that maximizes the effective capacity, with the maximum capacity and power constraints. The closed-form expression for the optimal current control strategy is developed. Numerical results validate that the proposed strategy differs from traditional QoS-aware power control strategies for EM wave communication and is optimal compared to other approaches.

The remainder of this paper is organized as follows. Section II presents the MI communication system model with statistical QoS provisioning. Section III derives the closed-form expression for the current control strategy to maximize effective capacity. Simulation results in Section IV validate the superiority of the proposed current control scheme. Conclusion is drawn in Section V.

\section{system Model}

As shown in Fig.~\ref{Underwater MI networks and equivalent circuits.}, we consider a typical point-to-point MI communication link in an underwater MI network and establish a statistical QoS provisioning framework for the link. In this model, the underwater receiver and the MI transmitter on the ship are each equipped with a coil for signal transmission and reception. When the upper layer of the transmitter sends packets to the lower layer, they are queued in a first-in-first-out (FIFO) sequence before being transmitted through the physical layer to the receiver.

The operating frequency is denoted by $f$. The current and voltage across the current source of the transmitter are denoted as $I_t$ and $U_t$, respectively. The receiver obtains the signal through mutual inductance and uses the voltage across the load resistor $R_L$ as the receive signal. The resistance, self-inductance, and the resonant capacitance corresponding to the transmit coil are $R_t$, $L_t$, and $C_t$, respectively. The current in the receive circuit is denoted as $I_r$, and the resistance, self-inductance, and the resonant capacitance corresponding to the receive coil are $R_r$, $L_r$ and $C_r$, respectively, where $j2\pi fL_t=\frac{1}{j2\pi fC_t}$, $j2\pi fL_r=\frac{1}{j2\pi fC_r}$. The mutual inductance between the transmit and receive coils is denoted as $M$. 

\subsection{Stochastic Misalignment Model for MI Links}
\label{2.A}
MI communication has the advantage of stable channel conditions. When two coils are parallel and aligned with each other and the distance between them is relatively small, we can calculate the mutual inductance between the two coils using Stokes' theorem\cite{frankl1986electromagnetic}:

\begin{equation}
    M_{\max}=\mu \pi N_tN_r\frac{a_{t}^{2}a_{r}^{2}}{2d^3},
\end{equation}
where $N_t$ is the number of turns of the transmit coil, $N_r$ is the number of turns of the receive coil, $a_t$ is the radius of the transmit coil, $a_r$ is the radius of the receive coil, and $d$ is the distance between the centers of the two coils.

The authors of \cite{7794767} proposed a random misalignment model, which assumes that the coil orientations are uniformly distributed, and it derives the mutual inductance PDF for practical configurations consisting of a single coil and collinear coil arrays. From the PDF, the expected attenuation due to misalignment can be determined, and an alignment factor $J$ is defined as
\begin{equation}
    M=M_{\max}\cdot J.
\end{equation}

In an underwater MI network, the transmit coil on the ship is often large. For portability, the receive coil of the MI receiver is much smaller than that of the transmitting coil. Under such condition, it can be assumed that the receive coil is placed within a uniform magnetic field generated by the transmit coil, the alignment factor for the mutual inductance between the two coils can be expressed as:
\begin{equation}
    f_J(J)=\left\{\begin{array}{ll}\frac{3}{2}(1-\frac{4}{\pi}\arccos\frac{|J|}{\sqrt{1-J^2}}),&\frac{1}{\sqrt3}\leq|J|<\frac{1}{\sqrt2},\\\frac{3}{2},&\frac{1}{\sqrt2}\leq|J|\leq1,\\0,&\text{otherwise.}\end{array}\right.
\end{equation}

Therefore, we can derive the probability density function $f_M(M)$ of $M$ using basic probability theory.

\subsection{Signal Model}

We denote the symbol to be transmitted as $\mathbf{X}=\left[ X_1,X_2,\cdots ,X_n \right] ^T$, and $X_k\sim \mathcal{C} \mathcal{N} (0,1), k=1,2,\cdots ,n$. After the digital-to-analog conversion (DAC), assuming the signal from the source is a sine wave, the current of the transmit loop when transmitting the \textit{k}-th symbol is given by
\begin{equation}
    I_{tk}=\mathfrak{R} \left\{ X_kI_te^{j2\pi ft} \right\}. 
\end{equation}

Then, based on Kirchhoff's voltage law, we obtain the following equations:
\begin{gather}
    U_{tk}=I_{tk}R_t-I_{rk}jwM,
    \\
    I_{tk}jwM=I_{rk}\left( R_L+R_r \right). 
\end{gather}

The receive current and receive voltage can be derived as follows:
\begin{align}
    I_{rk}&=\left\{ I_{tk}\frac{jwm}{R_L+R_r} \right\}, 
    \\
    U_{Lk}&=I_{rk}\times \,\,R_L+n_0
    \notag\\
    &=\mathfrak{R} \left\{ X_kI_{tk}e^{j2\pi ft}\frac{jwmR_L}{R_L+R_r} \right\} +n_0,
\end{align}
where $n_0$ represents the independent and identically distributed (i.i.d.) Gaussian white noise in the communication system, with zero mean and variance $\sigma_0^2$. After analog-to-digital conversion (ADC), the baseband signal can be expressed as
\begin{equation}
    y_k=X_kI_t\frac{jwmR_L}{R_L+R_r}+\tilde{n}_0,
\end{equation}
where $\tilde{n}_0$ represents the independent and identically distributed (i.i.d.) Gaussian white noise in the communication system, with zero mean and variance $\tilde{\sigma}_0^2$.

\subsection{Power model}
Due to the mutual inductance between the transmit and receive circuits, the transmit power is not solely determined by the parameters of the transmit circuit. Circuit equivalence also needs to be considered, and transmit and receive circuits will each has an equivalent resistance. Based on the Kirchhoff's voltage law, we can obtain
\begin{align}
    U_{tk}&=I_{tk}\times \left( R_t+\frac{\left( 2\pi f \right) ^2m^2}{R_L+R_r} \right), 
    \\
    P_{tk}&=\frac{1}{2}\mathbb{E} \left[ \Re \left( U_{tk}I_{tk} \right) \right] =\frac{1}{2}{I_{tk}}^2\times \left( R_t+am^2 \right), 
\end{align}
where $a=\frac{\left( 2\pi f \right) ^2}{R_L+R_r}$. The difference of the transmit power and the receive power varies with the mutual inductance. The relationships can be expressed as:
\begin{align}
    I_r&=I_t\frac{j\cdot 2\pi fM}{R_L+R_r},
    \\
    P_L&=\frac{1}{2}\mathbb{E} \left[ \Re \left( I_r^2R_L \right) \right] =\frac{1}{2}bI_{t}^{2}\tilde{\sigma}_0^2M^2,
\end{align}
where $b=\frac{\left( 2\pi f \right) ^2R_L}{\left( R_L+R_r \right) ^2\tilde{\sigma}_0^2}$.

\subsection{Statistical QoS Provisioning and Effective Capacity Model}

Although MI communication is regarded as a promising solution for underwater communication, improving the system throughput while ensuring user QoS remains a challenge. In EM wave communication, statistical QoS provisioning have been proven to be an effective and powerful technique for describing and implementing QoS provisioning that bind the delay of wireless real-time traffic\cite{1210731,9200891,8485468,4413145,10494937}. 

In mobile wireless networks, the authors of \cite{4290047} proposed a power and rate adaptive control scheme on a EM wave network. By combining information theory with the concept of effective capacity, the goal of this scheme is to maximize the system throughput subject to a given delay QoS constraint.

For a stationary and ergodic random service process, the probability that the queue length exceeds a threshold $Q_{th}$ decreases exponentially fast as $Q_{th}$ increases, which can be expressed as

\begin{equation}
    P\left( Q>Q_{th} \right) =e^{-\theta Q_{th}}.
\end{equation}

Based on the Large Deviation Principle, the authors of \cite{284868} demonstrate that under sufficient conditions, the queue length process converges in distribution to a random variable, such that
\begin{equation}
    -\lim_{Q_{th}\rightarrow \infty} \frac{\log \left( \mathrm{Pr}\{Q(\infty )>Q_{th}\} \right)}{Q_{th}}=\theta. 
\end{equation}

Here, $\theta$ is called the QoS exponent, which reflects the decay rate of the probability that the queue length exceeds a threshold. The larger value of $\theta$,  indicates that the service process demands a strict QoS, with a very low probability of queue overflow. On the other hand, if $\theta$ is small, it implies that the service process has a relaxed QoS, allowing for a higher probability of queue overflow. In particular, when $\theta \rightarrow \infty$, the service process cannot tolerate any probability of queue overflow, essentially enforcing the zero overflow probability. In contrast, when $\theta \rightarrow 0$, the service process can tolerate an arbitrarily large probability of queue overflow, essentially imposing no strict QoS requirement.

Based on effective bandwidth theory, the authors of \cite{1210731}, provide a dual definition of effective capacity. They define it as the maximum constant arrival rate that a given service process can support while meeting specified statistical delay-bounded QoS requirements. This is expressed as
\begin{equation}
    EC\left( \theta \right) \triangleq -\frac{1}{\theta}\log \left\{ \mathbb{E} _{\gamma}\left[ e^{-\theta R\left( t \right)} \right] \right\} .
\end{equation}
where $\mathbb{E} _{\gamma}\left[ \cdot  \right] $ represents the expected value of the SNR at the receiver. Let the sequence $\left\{ R\left[ i \right], \, i=1,2,\cdots \right\}$ represent the discrete, stationary, and ergodic random service process from the transmission circuit to the reception circuit, and let $S\left[ t \right] \triangleq \Sigma _{i=1}^{t}R\left[ i \right] $ denote the sum of the first $t$ terms of $R[i]$. Assuming that the Gartner-Ellis limit of $S(t)$ exists and is convex and differentiable, the effective capacity $E_C(\theta)$ is defined as
\begin{equation}
    S\left[ t \right] \triangleq \Sigma _{i=1}^{t}R\left[ i \right] .
\end{equation}

When the sequence $\left\{ R\left[ i \right], i=1,2,\cdots \right\} $ is an uncorrelated sequence, the above expression can be simplified to
\begin{equation}
    \mathrm{E}_{\mathrm{C}}(\theta )=-\frac{1}{\theta}\log \Bigl( \mathbb{E} _{\gamma}\left\{ e^{-\theta R[i]} \right\} \Bigr) ,
\end{equation}

\text

For the service process, we use the Shannon capacity to represent the maximum achievable data rate, assuming that both the transmitter and receiver have suitable encoding and decoding schemes that allow them to reach the channel capacity.
However, due to the hardware limitations of the transmitter, the maximum achievable rate is $R_{\max}$. For an MI system, the SNR at the load resistance of the receiving circuit is typically represented as
\begin{equation}
    \gamma =\frac{P_L}{\tilde{\sigma}_{0}^{2}}=\frac{1}{2}bI_{tk}^{2}M^2.
\end{equation}

The channel capacity can be represented as $ R\left( M \right) =\log _2\left( 1+\gamma \right) =\log _2\left( 1+bM^2I_{tk}^{2}\left( M \right) \right) $, the effective capacity  can be expressed as
\begin{align}
    &E_C\left( \theta \right) 
    \notag\\
    =&-\frac{1}{\theta}\log \left( \mathbb{E} _{M}\left[ \exp \left( -\theta R\left( M \right) \right) \right] \right) 
    \notag\\
    =&-\frac{1}{\theta}\log \left( \int_{-M_{\max}}^{M_{\max}}{\left( 1+\frac{1}{2}bM^2I_{tk}^{2}\left( M \right) \right) ^{-\beta}p_M\left( M \right) dM} \right) ,
\end{align}
where $\beta =\frac{\theta}{\ln 2}$. Under the assumption of a stationary and ergodic random sequence, for clarity, the time index $k$ will be omitted in the following sections of the paper to simplify the description.

\section{Optimal Current Control for QoS-aware Based MI Communication System}

Based on the system model described in Section II, the goal is to develop the optimal current control strategy in MI communication to maximize the effective capacity.

Because of the hardware limitations of the transmitter, we assume that the transmit current satisfies the peak current constraint as
\begin{equation}
    I_{tk}\leqslant I_{\max}\left( M \right) =\sqrt{\frac{2\left( 2^{R_{\max}}-1 \right)}{bM^2}},
\end{equation}
and the transmit power satisfies the average power constraint as 
\begin{equation}
    \mathbb{E} _{M}\left[ \frac{1}{2}{I_t}^2(M)\times \left( R_t+aM^2 \right) p_M\left( M \right) \right] \le \bar{P}.
\end{equation}

For convenience, let $\xi \left( M \right) =I_{t}^{2}\left( M \right) $. The objective function can then be formulated as follows:
\begin{align}
    \max -\frac{1}{\theta}\log \left( \mathbb{E} _M\left[ \left( 1+\frac{1}{2}bM^2\xi \left( M \right) \right) ^{-\beta}p_M\left( M \right) \right] \right) .
\end{align}

Since $-1/\theta \cdot \log \left( x \right) $ is a monotonically decreasing function of $x$, the objective function can be rewritten accordingly. we can express the optimization problem as follows:
\begin{align}
    \min & \int_{-M_{\max}}^{M_{\max}}{\left( 1+\frac{1}{2}bM^2\xi \left( M \right) \right) ^{-\beta}p_M\left( M \right) dM}
    \\
    s.t. & \int_{-M_{\max}}^{M_{\max}}{\frac{1}{2}\xi \left( M \right) \times \left( R_t+aM^2 \right) p_M\left( M \right) dM}\le \bar{P}
    \label{26}
    \\
     &   \xi(M)\leqslant\frac{2\left( 2^{R_{\max}}-1 \right)}{bM^2}
    \\
    &    \xi \left( M \right) \ge 0
\end{align}

\begin{theorem}
    The optimal current control scheme which is the solution to the optimization problem, denoted by $\xi _{opt}\left( M \right)$, is given by:
    \begin{equation}
        \xi _{opt}\left( M \right) =\begin{cases}
            0\text{,}&		0\le \left| M \right|\\
            &       \le M_1,\\
            \frac{2}{bM^2}\left\{ \left[ \frac{1}{\lambda _0}\times \frac{M^2}{R_t+aM^2} \right] ^{\frac{1}{\beta +1}}-1 \right\} ,&		M_1<\left| M \right|\\
            &		<M_2,\\
            \frac{2}{bM^2}\left( 2^{R_{\max}}-1 \right) ,&		M_2\le \left| M \right|\\
            &		\le M_{\max},\\
        \end{cases}
    \end{equation}
    where $M_1=\sqrt{\frac{R_t}{\frac{1}{\lambda _0}-a}}$, $M_2=\sqrt{R_t/\left( \frac{1}{\lambda _02^{R_{\max}\left( \beta +1 \right)}}-a \right)}$, $\lambda _0=\lambda _1/\left( b\beta \right) $, $\lambda_0$ can be solving numerically.
\end{theorem}

\begin{proof}
    Since this is a convex optimization problem, we can use the Karush-Kuhn-Tucker (KKT) conditions to derive the optimal current strategy. The Lagrangian function corresponding to the optimization problem can be written as
    \begin{align}
        &L\left( \xi \left( M \right) , \lambda _1, \lambda _2, \lambda _3 \right) =
        \notag\\
        &\int_{-M_{\max}}^{M_{\max}}{\left( 1+\frac{1}{2}bM^2\xi \left( M \right) \right) ^{-\beta}p_M\left( M \right) dM}+
        \notag\\
        &\lambda _1\int_{-M_{\max}}^{M_{\max}}{\frac{1}{2}\xi \left( M \right) \times \left( R_t+aM^2 \right) p_M\left( M \right) dM}+
        \notag\\
        &\lambda _2\log _2\left( 1+\frac{1}{2}bM^2\xi \left( M \right) \right) -\lambda _3\xi \left( M \right) .
    \end{align}

    The KKT conditions corresponding to the optimization problem can be written as follows:
    \begin{equation}
        \begin{cases}
            \frac{\partial L\left( \xi \left( M \right) , \lambda _1, \lambda _2 \right)}{\partial \xi \left( M \right)}=0,\\
            \lambda _1\left[ \int_{-M_{\max}}^{M_{\max}}{\frac{1}{2}\xi \left( M \right) \times \left( R_t+aM^2 \right) p_M\left( M \right) dM}-\bar{P} \right] =0,\\
            \lambda _2\left[ \log _2\left( 1+\frac{1}{2}bM^2\xi \left( M \right) \right) -R_{\max} \right] =0,\\
            -\lambda _3\xi \left( M \right) =0\\
            \lambda _1, \lambda _2,\lambda _3\ge 0,
        \end{cases}
    \end{equation}
    Where $\lambda_1$, $\lambda_2$ and $\lambda_3$ is the Lagrange multiplier.

    When $\lambda_2>0$, we have $\lambda_3=0$, because in this case $\xi(M)=0$, additionally, we also have $\log _2\left( 1+\frac{1}{2}bM^2\xi \left( M \right) \right) =0$, which leads to the solution for $\xi \left( M \right) =2R_{\max}/\left( R_t+aM^2 \right) $; When $\lambda_3>0$, we have $\lambda_2=0$; Since both $\lambda_1$ and $\lambda_2$ cannot be simultaneously greater than zero, then $\lambda_1=\lambda_2=0$, by solving the KKT conditions, the following can be obtained as:
    \begin{equation}
        \xi _{opt}\left( M \right) =\frac{2}{bM^2}\left\{ \left[ \frac{1}{\lambda _0}\times \frac{M^2}{R_t+aM^2} \right] ^{\frac{1}{\beta +1}}-1 \right\}.
    \end{equation}
\end{proof}

Let 
\begin{align}
    \begin{cases}
        \xi _1\left( M \right) =\frac{2}{bM^2}\left\{ \left[ \frac{1}{\lambda _0}\times \frac{M^2}{R_t+aM^2} \right] ^{\frac{1}{\beta +1}}-1 \right\},\\
        \xi _2\left( M \right) =\frac{2}{bM^2}\left( 2^{R_{\max}}-1 \right),\\
    \end{cases}
\end{align}
the effective capacity under the optimal current control strategy can be obtained as:
\begin{align}
    &E_C\left( \theta ,\xi \left( M \right) \right) =
    \notag\\
    &-\frac{2}{\theta}\log \int_{M_1}^{M_2}{2\left( 1+\frac{1}{2}bM^2\xi _1\left( M \right) \right) ^{-\beta}p_M\left( M \right) dM}+
    \notag\\
    &-\frac{2}{\theta}\log \int_{M_2}^{M_{\max}}{2\left( 1+\frac{1}{2}bM^2\xi _2\left( M \right) \right) ^{-\beta}p_M\left( M \right) dM}.
\end{align}

In particular, when $\theta \rightarrow 0$, the communication system can tolerate arbitrarily large queuing delays, and the optimal current control strategy reduces to
\begin{equation}
    \label{38}
    \xi _{opt}\left( M \right) =\begin{cases}
        0,&		0\le \left| M \right|\le M_1,\\
        \frac{2}{b}\left\{ \frac{1}{\lambda _0}\times \frac{1}{R_t+aM^2}-\frac{1}{M^2} \right\} ,&		M_1<\left| M \right|<M_2,\\
        \frac{2}{bM^2}\left( 2^{R_{\max}}-1 \right)&		M_2\le \left| M \right|\\
        &       \le M_{\max},\\
    \end{cases}
\end{equation}
where $M_1=\sqrt{\frac{R_t}{\frac{1}{\lambda _0}-a}}$, $M_2=\sqrt{R_t/\left( \frac{1}{\lambda _02^{R_{\max}}}-a \right)}$. 

In this case, the system can tolerate arbitrarily long queue delays, because $R_t\gg aM^2$, the optimal current control strategy converges to the channel water-filling theorem under MI communication.

When $\theta \rightarrow \infty$, the communication system cannot tolerate any non-zero delay. In this case, $M_1\rightarrow 0$, $M_2\rightarrow \infty$, and the system is not affected by the peak rate limitation, with the outage probability being zero. The optimal current control strategy is reduced to
\begin{equation}
    E_C\left( \theta ,\xi \left( M \right) \right) =\int_{M_1}^{M_2}{2\left( 1+\frac{1}{2}bM^2\xi _1\left( M \right) \right) ^{-\beta}p_M\left( M \right) dM}.
\end{equation}

\section{Simulation Result}

We conduct simulations to evaluate the performance of the optimal current control strategy that we propose in MI communication. Throughout our simulations, we set the operating frequency $f=2~\mathrm{kHz}$, the transmission distance $d=3~\mathrm{m}$, the loop resistance $R_t=R_r=20~\Omega$, the load resistance $R_L=200~\Omega$, the radius of the coils $a_t=4~\mathrm{m}$, $a_r=0.25~\mathrm{m}$, the number of turns $N_t=N_r=20$, the maximum normalized achievable rate $R_{\max}=0.5~\mathrm{bps/Hz}$, the average transmit power $\bar{P}=9~\mathrm{W}$, and the noise power $\tilde{\sigma}_{0}^{2}=2.5\times10^{-3}~\mathrm{W}$. 

\begin{figure}[htbp]
    \centering
    \includegraphics[width = 3.5in]{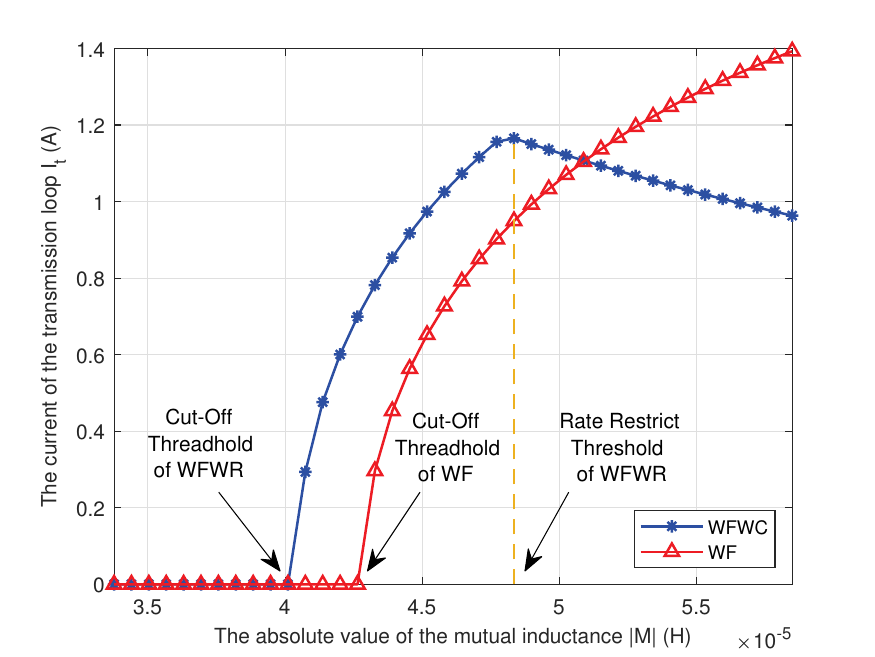}
    \caption{The transmit currents comparison between the WFWC and the WF strategy.}
    \label{The transmit currents comparison between the WFWR and the WF strategy.}
\end{figure}
    
Figure~\ref{The transmit currents comparison between the WFWR and the WF strategy.} shows the water-filling (WF) strategies with and without peak current constraint (WFWC) when arbitrarily long delays can be tolerated ($\theta=0$). Compared with the WF strategy, it can be observed that imposing a current rate constraint introduces the current rate constraint threshold. When $|M|$ exceeds the peak rate constraint threshold, the current decreases to reduce energy consumption while maintaining a constant achievable rate. When comparing the cut-off thresholds of the two strategies, it is evident that the cut-off threshold of WFWR shifts to the left, indicating that WFWR uses power more efficiently by saving it for when the current is lower. This leads to a reduction of approximately $11\%$ in the probability of outage, thus improving the stability of information transmission.
\begin{figure}[htbp]
    \centering
    \includegraphics[width = 3.5in]{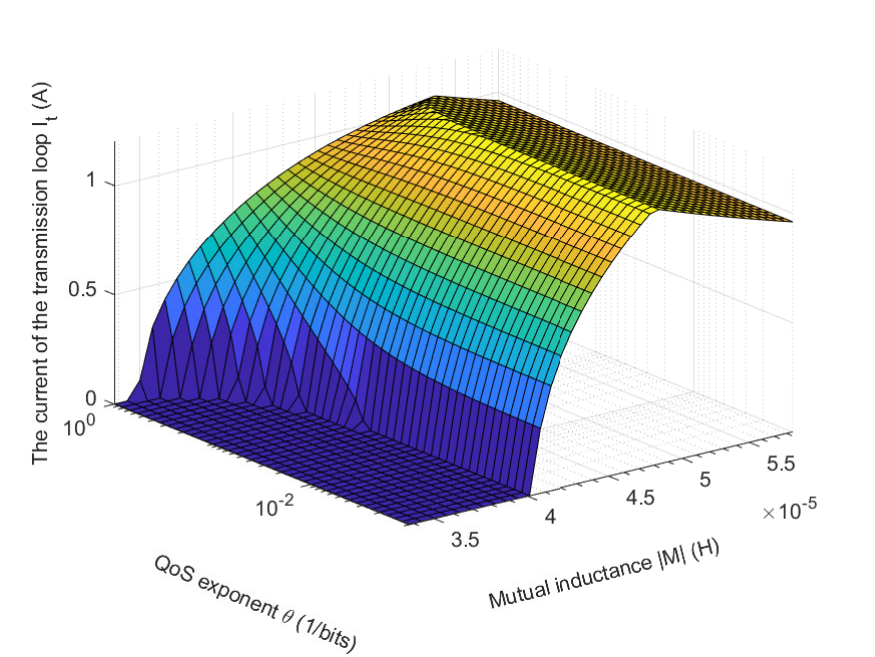}
    \caption{The optimal current adaptation strategy with current constraint.}
    \label{The optimal current adaptation strategy with rate restricted.}
\end{figure}

Figure~\ref{The optimal current adaptation strategy with rate restricted.} illustrates the optimal current control strategy under different statistical as $\theta$ increases, until it becomes smaller than $M_{\min}$, at which point it no longer constrains the system. The peak current constraint threshold increases with the increase of $\theta$, until it becomes larger than $M_{\max}$. When $\theta\rightarrow\infty$, the current control strategy becomes a continuously differentiable function. As $\theta$ decreases, the optimal current control strategy converges to WFWR.

\begin{figure}[htbp]
    \centering
    \includegraphics[width = 3.5in]{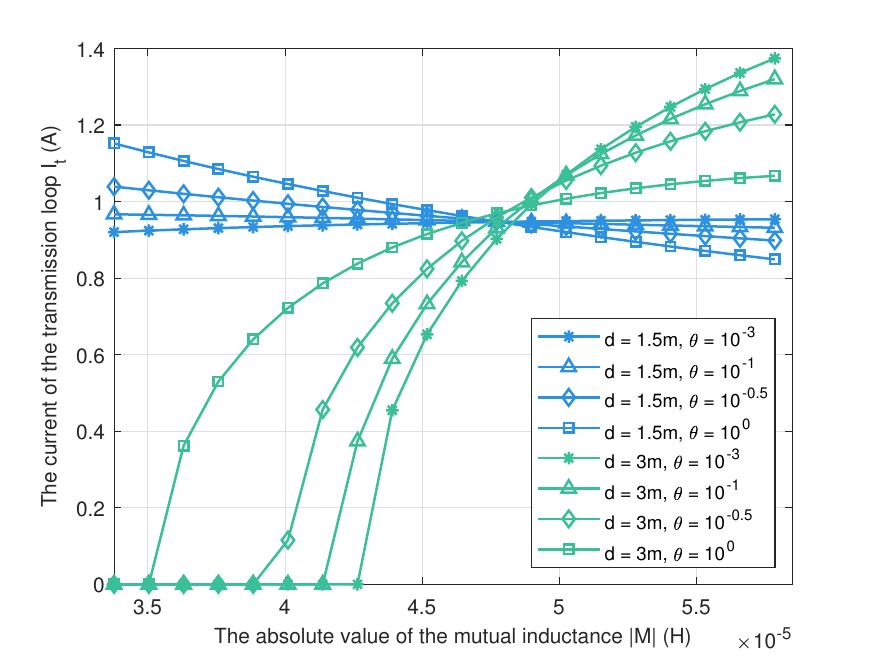}
    \caption{The optimal power-adaptation policy without current restriction.}
    \label{The optimal power-adaptation policy without rate restriction.}
\end{figure}

Figure~\ref{The optimal power-adaptation policy without rate restriction.} demonstrates the optimal current control strategy under different distances and statistical QoS requirements, without the peak power constraint. It is easy to observe that, although the closed-form expressions for different distances are the same, the actual current control differs. This is particularly evident for short distances, and $\theta\rightarrow 0$, the current control strategy approximates a constant current distribution.
\begin{figure}[htbp]
    \centering
    \includegraphics[width = 3.5in]{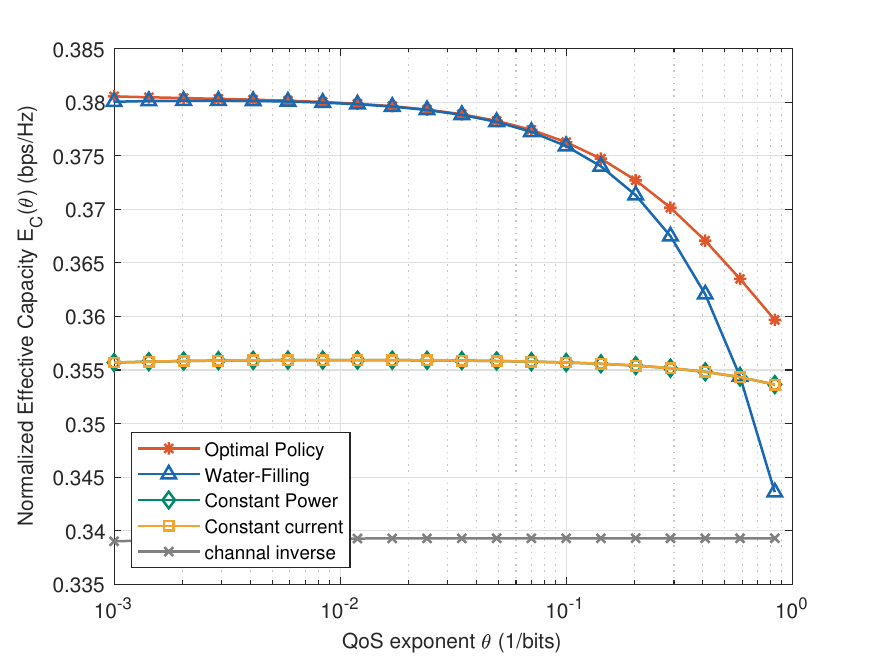}
    \caption{Normalized effective capacities under our optimal current control strategy with other classic strategies.}
    \label{Normalized effective capacities under our optimal current control strategy with other classic strategies.}
\end{figure}

The comparison of the normalized effective capacity for different current adaptive allocation schemes is shown in Fig.~\ref{Normalized effective capacities under our optimal current control strategy with other classic strategies.}. When $\theta$ is small, the optimal current adaptive control strategy resembles the water-filling strategy. However, as $\theta$ increases, it does not converge to the channel inversion current control strategy. As the statistical QoS requirement becomes more stringent, the normalized effective capacity of all current control schemes decreases. However, the proposed current adaptive control scheme  consistently achieves the highest effective capacity.

\section{Conclusion}
In the context of underwater MI networks, we apply MI communication technology to establish a communication system with statistical QoS provisioning. Mathematical expressions for transmit power, receive power, and effective capacity are derived. Based on the statistical QoS provisioning theory, we consider an optimization problem to maximize effective capacity, with constraints on average power and peak current. The closed-form expression for the optimal current control strategy is derived. Finally, through numerical simulations, the effects of the QoS exponent $\theta$ and the communication distance $d$ on the optimal current control strategy are verified, demonstrating that the proposed strategy can significantly improve the network throughput while provisioning delay QoS requirements.

\bibliographystyle{IEEEtran}
\bibliography{references.bib}

\end{document}